\documentclass{article}
\usepackage[spanish,english]{babel}
\usepackage{amssymb,amsfonts,amsmath}
\usepackage {color}
\usepackage[left,pagewise] {lineno}
\usepackage {xspace}
\usepackage {url}
\usepackage {plaatjes}

\newcommand {\spacedvijfplaatjes} [7] []
{
  \plaatjes
  {
    \subplaatje {#1} {#2}
    \subplaatje {#1} {#3}
    \subplaatje {#1} {#4}\\
    \hspace*{1.3cm}
    \subplaatje {#1} {#5}
    \hspace*{-0.5cm}
    \subplaatje {#1} {#6}
    \hfill{}
  } {#7} {#2+#3+#4+#5+#6}
}

\newcommand{\niceremark}[3]{\textcolor{red}{\textsc{#1 #2: }}\textcolor{blue}{\textsf{#3}}}
\newcommand{\mati}[2][says]{\niceremark{Mati}{#1}{#2}}
\newcommand{\maarten}[2][says]{\niceremark{Maarten}{#1}{#2}}
\newcommand{\rodrigo}[2][says]{\niceremark{Rodrigo}{#1}{#2}}

\newcommand{\mkmrm}  [1]{\ensuremath{\mathrm{#1}}\xspace}
\newcommand{\mkmbb}  [1]{\ensuremath{\mathbb{#1}}\xspace}

\newcommand{\mkmfrak}[1]{\ensuremath{\mathfrak{#1}}\xspace}

\newcommand {\R} {\mkmbb {R}}

\newcommand {\etal} {\textit {et al.}}
\newcommand {\eps} {\varepsilon}
\newcommand{\alp}{\mkmfrak{A}}
\newcommand{\sym}[1]{\mkmfrak{#1}}
\newcommand{\ply}{\Delta}
\newcommand{\shpa}{\mkmrm{s\hspace{-2pt}p}}
\newcommand{\cost}{\mkmrm{cost}}

\renewcommand{\paragraph}[1]{\medskip\noindent\textbf{#1.}}

\newcommand{\regions}{\ensuremath{\cal D}}
\newcommand{\arr}{\ensuremath{\cal A(\regions)}}
\newcommand{\res}{\ensuremath{r}}
\newcommand{\thk}{\ensuremath{t}}

\definecolor {infocolor} {rgb} {0.6,0.6,0.6}

\newtheorem{observation}{Observation}
\newtheorem{theorem}{Theorem}
\newtheorem{corollary}{Corollary}
\newtheorem{lemma}{Lemma}
\newenvironment {proof}{\textbf {Proof:}}{\hfill \ensuremath {\boxtimes}}

\title{On the Complexity of Barrier Resilience\\for Fat Regions and Bounded Ply\footnote{A preliminary version of this work appeared at the 9th International Symposium on Algorithms and Experiments for Sensor Systems, Wireless Networks and Distributed Robotics~\cite{klss-otcbrfr-13}.}}

\author{
  Matias Korman\footnote
  { Tohoku University, Japan.
  {\tt mati@dais.is.tohoku.ac.jp}. 
    }
    \and 
  Maarten L\"offler\footnote
  { Dept. of Computing and Information Sciences, Utrecht University, the Netherlands. 
    \texttt{m.loffler@uu.nl}
  }
   \and 
  Rodrigo I. Silveira\footnote{
    Dept. de Matem\`atiques, Universitat Polit\`ecnica de Catalunya, Spain. 
    \texttt{rodrigo.silveira@upc.edu}
  }
   \and 
  Darren Strash\footnote
  { Department of Computer Science, Colgate University, USA.
    \texttt{dstrash@cs.colgate.edu}
  }
}

\begin{document}

\date{}

\maketitle

\begin{abstract}
In the \emph {barrier resilience} problem (introduced by Kumar {\em et al.}, Wireless Networks 2007), we are given a collection of regions of the plane, acting as obstacles, and we would like to remove the minimum number of regions so that two fixed points can be connected without crossing any region. In this paper, we show that the problem is NP-hard when the collection only contains fat regions with bounded ply $\ply$ (even when they are axis-aligned rectangles of aspect ratio $1 : (1 + \eps)$). We also show that the problem is fixed-parameter tractable (FPT) for unit disks and for similarly-sized $\beta$-fat regions with bounded ply $\ply$ and $O(1)$ pairwise boundary intersections.
 We then use our FPT algorithm to construct an 
$(1+\eps)$-approximation algorithm that runs in
$O(2^{f(\ply, \eps,\beta)}n^5)$
time, where $f\in O(\frac{\ply^4\beta^8}{\eps^4}\log(\beta\ply/\eps))$.
\end{abstract}

\section {Introduction}
The \emph {barrier resilience} problem asks for the minimum number of spatial regions from a collection $\regions$ that need to be removed, such that two given points $p$ and $q$ are in the same connected component of the complement of the union of the remaining regions.
This problem was posed originally in 2005 by Kumar \etal~\cite {kla-bcws-05,kla-bcws-07}, motivated from sensor networks. In their formulation, the regions are unit disks (sensors) in some rectangular strip $B \subset \R^2$, where each sensor is able to detect movement inside its disk. The question is then how many sensors need to fail before an entity can move undetected from one side of the strip to the opposite one (that is, how \emph{resilient} to failure the sensor system is).
Kumar \etal\ present a polynomial time algorithm to compute the resilience in this case.
They also consider the case where the regions are disks in an annulus, but their approach cannot be used in that setting.

\subsection {Related Work}

Despite the seemingly small change from a rectangular strip to an annulus, the second problem still remains open, even for the case in which regions are unit disks in $\R^2$.
There has been partial progress towards settling the question:
Bereg and Kirkpatrick~\cite {bk-abrwsn-09} present a factor $5/3$-approximation algorithm for the unit disk case. This result was very recently improved to a 1.5-approximation by Chan and Kirkpatrick~\cite{ck-mpamcppaabr-14}.  On the negative side, Alt \etal~\cite {acgk-cpslsa-11}, Tseng and Kirkpatrick~\cite {tk-brsn-11}, and Yang~\cite[Section 5.1]{y-sppacgatm-12} independently showed that if the regions are line segments in $\R^2$, the problem is NP-hard. Tseng and Kirkpatrick~\cite{tk-brsn-11} also sketched how to extend their proof for the case in which the input consists of (translated and rotated) copies of a fixed square or ellipse.

The problem of covering barriers with sensors has received a lot of attention in the sensor network community (e.g.,~\cite {c-kbcm-12,cglw-ammsm-12,hclss-cebc-12}).
In the algorithms community, closely related problems involving region intersection graphs have also become quite popular.
Gibson~\etal~\cite {gkv-ipud-11} study a problem that is, in a sense, opposite of ours: given a set of points and disks separating them (i.e., every path between two points intersects some disk), compute the maximum number of disks one can remove while keeping the points separated. They present a constant-factor approximation algorithm for this problem. Later, Penninger and Vigan showed that the problem is NP-complete~\cite{pv-psiuudnp-13}.
Recently, Cabello and Giannopoulos~\cite{Cabello2016} gave a cubic-time algorithm for the case where only two points have to be kept separated, for barriers that are arbitrary connected curves (under some mild assumptions).


\subsection {Results}
We present constructive results for two natural restricted variants of the problem. In Section~\ref {sec:fpt} we show that the problem is fixed-parameter tractable on the resilience when the regions are unit disks. We then extend this approach to other shapes that resemble unit disks. This resemblance is measured with the following three restrictions: all regions are of similar size, region boundaries have $O(1)$ pairwise intersections, and the collection of regions have bounded {\em ply}~\cite {mttv-sspnng-92} (that is, no point of the plane is covered by too many sensors). Such restrictions are similar in spirit to previous results that bound the union complexity of fat (and non-fat) regions~\cite{deberg-2008,efrat-05,whitesides-zhao-1990}. Formal definitions of fatness, ply, and more detailed descriptions of our restrictions are given in Section~\ref{sec:fpt-fat}.
 In Section~\ref {sec:approx} we also show that
the FPT result can be used to obtain an approximation scheme. In particular, the constructive results apply to the original unit disk coverage setting when the collection of disks (or in general fat objects) has bounded ply.

As a complement to these algorithms, in Section~\ref {sec:np} we show that the problem is NP-hard even when the input is a collection of fat regions of arbitrary shape in $\R^2$. The result holds even if regions consist of axis-aligned rectangles of aspect ratio $1:1+\eps$ and $1+\eps:1$. 
%
 Our results rely on tools and techniques from both computational geometry and graph theory. 



\section{Preliminaries} 
We denote with $p$ and $q$ the points that need to be connected, and with $\regions$ the set of regions that represent the sensors. 
To simplify the presentation of our results, we make the following general position assumption: all intersections between boundaries of regions in $\regions$ consist of isolated points.
We say that a collection of objects in the plane are {\em pseudodisks} if the boundaries of any two of them intersect at most twice.

 We formally define the concepts of \emph{resilience} and \emph{thickness} introduced in~\cite{bk-abrwsn-09}. The \emph{resilience of a path} $\pi$ between two points $p$ and $q$, denoted $\res(\pi)$, is the number of regions of $\regions$ intersected by $\pi$.
Given two points $p$ and $q$, the \emph{resilience of $p$ and $q$}, denoted $\res(p,q)$, is the minimum resilience over all paths connecting $p$ and $q$.
In other words, the resilience between $p$ and $q$ is the minimum number of regions of $\regions$ that need to be removed to have a path between $p$ and $q$ that does not intersect any region of $\regions$. Note that sometimes we will assume that neither $p$ nor $q$ are contained in any region of $\regions$, since such regions must always be counted in the minimum resilience paths, hence we can ignore them (and update the resilience we obtain accordingly). 

Often it will be useful to refer to the arrangement (i.e., the subdivision of the plane into faces; see~\cite{bcko-aad-08} for a formal definition) induced by the regions of $\regions$, which we denote by $\arr$. 
Based on this arrangement we define a weighted dual graph $G_{\arr}$ as follows. 
There is one vertex for each face (i.e., 2-dimensional cell) of $\arr$. Each pair of neighboring cells $A, B$ is connected in $G_{\arr}$ by two directed edges, $(A,B)$ and $(B,A)$. The weight of an edge is $1$ if, when traversing from the starting cell to the destination one, we enter a region of $\regions$ (or $0$ if we leave a region\footnote{Note that no other option is possible under our general position assumption.}).
\eenplaatje
{graph-example}
{The graph $G_{\arr}$ for an arrangement of three disks. Solid edges have weight 1, dashed edges have weight 0.
}


The \emph{thickness} of a path $\pi$ between $p$ and $q$, denoted $\thk(\pi)$, equals the number of times $\pi$ enters a region of $\regions$ when traveling from $p$ to $q$ (possibly counting the same region multiple times).
Given two points $p$ and $q$, the \emph{thickness of $p$ and $q$}, denoted $\thk(p,q)$, 
is the value $|\shpa_{G_{\arr}} (p,q)|+\ply(p)$, where $\shpa_{G_{\arr}} (p,q)$ is a shortest path in $G_{\arr}$ from the cell of $p$ to the cell of $q$, and $\ply(p)$ equals the number of regions that contain $p$. Also note that the resilience (or thickness) between two points only depends on the cells to which the points belong. Hence, we can naturally extend the definitions of thickness to encompass two cells of \arr, or a cell and a point.
Unless otherwise stated, we will use $\rho$ to denote a path with minimum resilience, and $\tau$ for one of minimum thickness.

Note that thickness and resilience can be different (since entering the same region several times has no impact on the resilience, but is counted every time for the thickness). 
In fact, the thickness between two points can be efficiently computed in polynomial time using any shortest path algorithm for weighted graphs (for example, using Dijkstra's algorithm).  
 However, as we will see later, the thickness (and the associated shortest path) will help us find a path of low resilience. 

Throughout the paper we often use the following fundamental property of disks, already observed in~\cite{bk-abrwsn-09}.
In the statement below, ``well-separated'' is in the sense used in~\cite{bk-abrwsn-09}---that is, the distance between $p$ and $q$ is at least $2\sqrt{3}$.\footnote {Note that the well-separatedness of $p$ and $q$ is used to prove a factor 2 instead of 3. Everything still works for points that are not well-separated, at a slight increase of the constants. Our most general statements for $\beta$-fat regions do not make this requirement.}

\begin{lemma}[\cite{bk-abrwsn-09}, Lemma~1]
\label{lem:AtMostTwice}
Let $\regions$ be a set of unit disks, and let $\rho$ be a path from $p$ to $q$ of minimum resilience. 
If $p,q$ are well-separated, then $\rho$ encounters no disk of $\regions$ more than twice.
\end{lemma}

\begin{corollary}[\cite{bk-abrwsn-09}]
\label{cor:AtMostTwice}
When the regions of $\regions$ are unit disks,
the thickness between two well-separated points is at most twice their resilience.
\end{corollary}
\tweeplaatjes {nounit-size} {nounit-ply} 
{(a) With $n$ unit disks and one arbitrarily large disk (orange), the optimal tour may be forced to enter and leave the same region up to $\Omega(n)$ times, even when the ply of the disks is at most $3$.
 (b) When we move the yellow disks closer together, the radius of the orange disk can be made arbitrarily close to $1$, at the cost of increasing the ply (i.e., having many disks covering the same point).
}
Note that a crucial property in the above results is that all disks have the same size. In Figure~\ref{fig:nounit-size} we show problem instances with a single large disk that has to be traversed a linear number of times in any minimum resilience path. The same instance is then modified in Figure~\ref{fig:nounit-ply} so that the radius of the larger disk is only $1+\eps$ times larger than the radius of the other disks (at the expense of concentrating all disks at the same point).



\section{Fixed-parameter tractability} \label {sec:fpt}
In this section we introduce a single-exponential fixed-parameter tractable (FPT) algorithm, where the parameter is the resilience of the optimal solution. Thus, our aim is to obtain an algorithm that given a problem instance, determines whether or not there is a path of resilience $r$ between $p$ and $q$, and runs in $O(2^{f(r)}n^c)$ time for some constant $c$ and some polynomial function $f$.

For clarity we first explain the algorithm for the special case of unit disks. Afterwards, in Section \ref{sec:fpt-fat}, we show how to adapt the solution to the case in which $\regions$ is a collection of $\beta$-fat objects. Note that for treating the case of unit disk regions we assume that  $p$ and $q$ are well-separated, so we can apply Lemma~\ref{lem:AtMostTwice}. This requirement is afterwards removed in Section~\ref{sec:fpt-fat}.



First we give a quick overview of the method of Kumar \etal~\cite {kla-bcws-05} for open belt regions. Their idea consists of considering the intersection graph of $\regions$ together with two additional artificial vertices $s_a$,$t_a$ with some predefined adjacencies. There is a path from the bottom side to the top side of the belt if and only if there is no path between $s_a$ and $t_a$ in the graph. Hence, computing the resilience of the network is equal to finding a minimum vertex cut between $s_a$ and $t_a$.

We start by giving a bird's-eye view of our algorithm. Let $\rho$ be a path of minimum resilience from $p$ to $q$, and let $\pi$ be any known path that starts at $p$, passes through $q$, and reaches an unbounded region. Assume that somehow we know that $\rho$ and $\pi$ do not cross (other than at $p$ and $q$). Then, we can cut open through $\pi$ effectively splitting the regions of $\arr$ traversed by the path into two. Topologically speaking, we get something that is homeomorphic to an open belt region, and thus we can solve the problem as such: construct the intersection graph, connect the split regions of $\arr$ to either of the artificial vertices depending on which side of the cut they lie in, and look for a minimum vertex cut (see Figure~\ref{fig:fig_idea}, left). Note that, when doing this cut, it is possible that a disk is split into more than one component. Whenever this happens, we must identify the portions as one (i.e., when one portion is entered, then entering the other portions of the same disk is {\em for free}).

\eenplaatje [scale=1]
{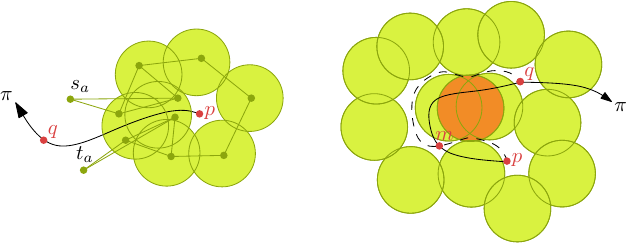}
{ (left) If we are given an infinitely long path $\pi$ (in black in the figure) that goes through $p$ and $q$, and is not crossed by $\rho$, we can cut open through it and obtain an open belt instance. The resulting graph (with the artificial vertices) is shown for clarity.
  (right) When the two paths intersect ($\rho$ denoted with a dashed path) we obtain several open belt problem instances. However, these problems are not independent, since the removal of the highlighted disk makes the paths from $p$ to $m$ and from $m$ to $q$ feasible.
}

Thus, the problem is easy once we have a path $\pi$ that does not cross with $\rho$. Unfortunately, finding such a path is difficult. Instead, we use several observations to compute a (possibly non-simple) path that cannot have many crossings with $\rho$, and guess where (if any) these crossings happen. Naturally, we don't know the way in which the two paths interact, but we will try all possibilities and return the one whose resulting resilience is smallest. A fixed crossing pattern decomposes $\rho$ into subpaths whose endpoints are in $\pi$ (see Figure~\ref{fig:fig_idea}, right). Although the subpaths are unknown, we can compute them via the usual open belt region approach. The main problem is that the different sub-problems are not independent (removing a single region may be useful for several subpaths). Thus, rather than finding a vertex cut that isolates the single source to the single sink, we are given a list of sources and sinks that need to be pairwise disconnected from each other. In the literature, this problem is known as the {\em vertex multicut problem}~\cite{x-sipamc-10}, and several FPT approaches are known. 



We now present some observations that will allow us to have a nice choice of $\pi$ (i.e., find a path in which the number of crossings with $\rho$ does not depend on $n$). Consider a minimum resilience path $\rho$ of shortest length between the cells containing $p$ and $q$ in $G_{\arr}$, and let $t$ be the number of disks traversed by $\rho$. Since $\rho$ has shortest length, it does not enter and leave the same region unless it helps reduce resilience. Since we assumed that $p$ is not contained in any region, $t$ is exactly the thickness of $p$ and $q$. We observe that cells with high thickness to $p$ or $q$ can be ignored when we look for low resilience paths.

\begin {lemma}
\label{lem:dist1.5}
The minimum resilience path $\rho$ between $p$ and $q$ cannot traverse cells whose thickness to $p$ or $q$ is larger than $1.5t$.
\end {lemma}
\begin {proof}
We argue about thickness to $p$; the argument with respect to $q$  is analogous. 
Let $\rho$ be a path of minimum resilience between $p$ and $q$, and let $r$ be the resilience of $\rho$. Also, let $\tau$ be a minimum-thickness path from $p$ to $q$. Recall that $\rho$ does not enter a disk more than twice, hence the thickness of $\rho$ is at most $2r\leq 2t$. Assume, for the sake of contradiction, that the thickness of some cell $C$ traversed by $\rho$ is greater than $1.5t$. Let $\rho_C$ be the portion of $\rho$ from $C$ to $q$. Since the thickness of $\rho$ from $p$ to $q$ is at most $2t$, the triangular inequality implies that the thickness of $\rho_C$ is less than $0.5t$.
 
Now, by concatenating $\tau$ and $\rho_C$. we would obtain a path that connects $p$ with $C$ whose thickness is less than $1.5t$, giving a contradiction with the thickness of cell $C$. 
\end {proof}

For simplicity in the exposition, we will also bound the region to consider (thus, we discard regions with very high resilience since they will not be traversed by $\rho$). Let $R$ be the union of the cells of the arrangement that have thickness from $p$ at most $1.5t$; we call $R$ the \emph {domain} of the problem. Observe that $R$ is connected, but need not be simple (see Figure~\ref{fig:fpt-path-Mati}).

For simplicity in the explanation, we add additional discs surrounding $\regions$ so as to make sure that the unbounded face has thickness more than $1.5t$. This does not affect the asymptotic behavior of our algorithm, but it removes the need of considering some degenerate situations. Note that the number of cells remaining in $R$ might still be quadratic, hence asymptotically speaking the instance size has not decreased (the purpose of this pruning will become clear later). 

\tweeplaatjes [scale=.65]
{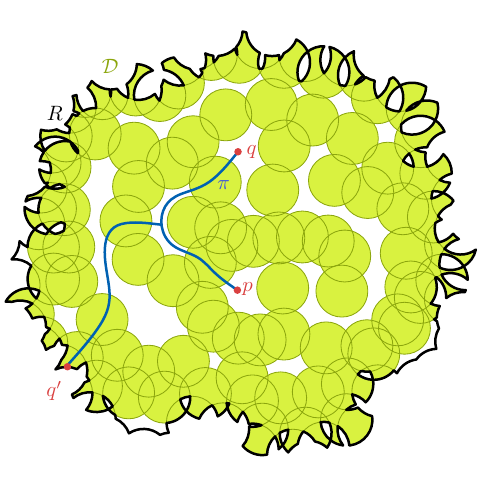}
{fpt-path}
{In order to transform our problem to one that resembles an open belt, we remove all cells of high thickness and cut through the tree formed by the union of two shortest paths. 
Figures (a) and (b) show two examples of the result.}

\tweeplaatjes [scale=.65]
{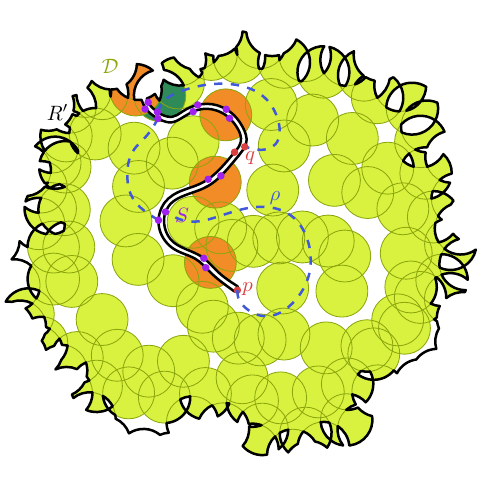}
{fpt-eatenModified}
{
  (a) Illustration of the situation in Figure~\ref{fig:fpt-path} after cutting along $\pi$.
  We get the domain $R'$, add a set $S$ of extra vertices on the boundary of $R'$, and end up with two copies of $q$. A crossing pattern, consisting of a topological path $\rho$ (defined by the sequence of points of $S$ it passes).
  The disks of $\regions$ intersected by $\pi$ are shown green if they are crossed by $\rho$, and orange otherwise.
  (b) Domain after removing the disks traversed by $\pi$ that are not crossed by $\rho$. 
  The green disk (shown transparent) is added to the solution, and thus ignored from now on.
}




\begin {lemma}\label{lem_pathtau}
There exists a point $q'$ on the outer boundary of $R$ and a tree that spans $p$, $q$, and $q'$ that has total thickness\footnote{The thickness of the tree is defined as the thickness of the paths that compose the tree.} $2.5t$.
\end {lemma}

\begin{proof}
Pick any point $q'$ in the outer boundary of $R$ and consider the tree obtained by joining the shortest paths from $q'$ to $p$, and $p$ to $q$. 
Note that the two paths may go through the same cell of $R$, see Figure~\ref{fig:fpt-path-Mati}.
The exact paths chosen are not important provided that they have no proper crossings. By definition, the thickness of each of these paths cannot exceed $1.5t$ and $t$, respectively, hence the lemma is shown.
%
\end{proof}

Let $\pi$ be the path from $q'$ to $q'$ that traverses the tree from the previous lemma.
We ``cut open'' through $\pi$, removing it from our domain. Note that cells that are traversed by $\pi$ are split into two copies (or three in the case of the cell containing $m$) of the same Jordan curve (See Figure~\ref{fig:fpt-pattern}). 

Consider now a minimum resilience path $\rho$, and let $r=\res(\rho)$ denote its resilience. This path can cross $\pi$ several times, and it can even coincide with $\pi$ in some parts (shared subpaths). Although we do not know how and where these crossings occur, we can \emph{guess} (i.e., try all possibilities) the topology of $\rho$ with respect to $\pi$. For each disk that $\pi$ passes through, we consider two cases: if $\rho$ goes through it, it will be part of the solution, and can be ignored from now on (increasing by one the total resilience). Otherwise, we make it an obstacle, removing it from the domain, see Figure~\ref{fig:fpt-eatenModified}. 
In that way we know the exact behavior of $\rho$ in the regions traversed by $\pi$. Additionally, we guess how many times $\rho$ and $\pi$ share part of their paths (either for a single crossing in one cell, or for a longer shared subpath). 
For each shared subpath, we guess from which cell $\rho$ arrives and leaves.

We call each such configuration a {\em crossing pattern} between $\pi$ and $\rho$. More formally, a single crossing is described by a tuple of four cells: the first cell $C$ that the two paths have in common for that crossing, the cell that $\rho$ visits right before entering $C$. Similarly, we add the last cell that the two paths have in common and the cell that is afterwards entered by $\rho$. A crossing pattern is described by a sorted list of all the crossings that $\pi$ and $\rho$ have. 

\begin{lemma}\label{lem_patterns}
For any problem instance $\regions$, there are at most $2^{4r\log r + o(r\log r)}$ crossing patterns between $\pi$ and $\rho$, where $r=\res(\rho)$.
\end{lemma}
\begin{proof}
First, for all disks in $\pi$, we guess whether or not they are also traversed by $\rho$. By Lemma~\ref{lem_pathtau}, $\pi$ has thickness at most $2.5t$, there are at most such many disks (hence up to $2^{2.5t}$ choices for which disks are traversed by $\rho$).

Now observe that $\pi$ cannot traverse many cells of $\arr$: when moving from a cell to an adjacent one, we either enter or leave a disk of $\regions$. Since we cannot leave a disk we have not entered and $\pi$ has thickness at most $2.5t$, we conclude that at most $5t$ cells will be traversed by $\pi$ (other than the starting and ending cells). 

We now bound the number of (maximal) shared subpaths between $\rho$ and $\pi$: recall that $\rho$ passes through exactly $r=\res(\rho)$ disks, and visits each disk at most twice. Hence, there cannot be more than $2 r$ shared subpaths. For each shared subpath we must pick two of the cells traversed in $\pi$ (as candidates for first and last cell in the subpath). By the previous observation there are at most $5t$ candidates for first and last cell (since that is the maximum number of cells traversed by $\pi$). Additionally, for each shared subpath we must determine from which side $\rho$ entered and left the subpath; in most cases we have two options for entering and leaving (since most cells are split into two by $\pi$). However, it could happen that the first, last (or even both cells) are the cell containing $m$. The cell containing $m$ was split into three, and thus we have three options on which part of the cell $\rho$ enters or leaves. That is, on the worst case there are three possibilities where $\rho$ enters and three possibilities where $\rho$ leaves the path, which gives a total of nine options overall. Since these choices are independent, in total we have at most $2r\times (5t \times 5t\times 9)^{2r}=101250^{r}\cdot t^{4r}r$ possibilities.

That is, in order to determine a crossing pattern, we must fix which disks of $\pi$ are traversed by $\rho$ as well as how many and where do the crossings between $\rho$ and $\pi$ happen. The bounds for each of these terms are $2^{2.5t}$ and $101250^{r}\cdot t^{4r}r$, respectively. Since these choices are independent, and using the fact that $t \leq 2r$, we obtain: 
\begin{eqnarray}
2^{2.5t} \cdot 101250^{r}\cdot t^{4r}r  &\leq &
2^{5r}\cdot 101250^{r}\cdot (2r)^{4r}r  \nonumber \\
& =& 2^{5r+r\log{101250} + 4r\log 2r+\log r} \nonumber \\
&=&2^{4r\log r + o(r\log r)}  \nonumber 
\end{eqnarray}
\end{proof}

Note that the bound is very loose, since most of the choices will lead to an invalid crossing pattern. 
However, the importance of the lemma is in the fact that the total number of crossing patterns only depends on $r$. 

Our FPT algorithm works by considering all possible crossing patterns, finding the optimal solution for a fixed crossing pattern, and returning the solution of smallest resilience. From now on, we assume that a given pattern has been fixed, and we want to obtain the path of smallest resilience that satisfies the given pattern. If no path exists, we simply discard it and associate infinite resilience to it.

\subsection{Solving the problem for a fixed crossing pattern}
Recall that the crossing pattern gives us information on how to deal with the disks traversed by $\pi$. Thus, we remove all cells of the arrangement that contain one or more disks that are forbidden to $\rho$. Similarly, we remove from $\regions$ the disks that $\rho$ must cross. After this removal, several cells of our domain may be merged. 

Since we do not use the geometry, we may represent our domain by a disk $W$
 (possibly with holes). After the transformation, each remaining region of $\regions$ becomes a pseudodisk, and $\rho$ becomes a collection of disjoint partial paths, each of which has its endpoints on the boundary of $W$ (see Figure~\ref{fig:fpt-schematic}), but is otherwise not yet fixed. To solve the subproblem associated with the crossing pattern  we must remove the minimum number of disks so that all partial paths are feasible.

\spacedvijfplaatjes [scale=.8] {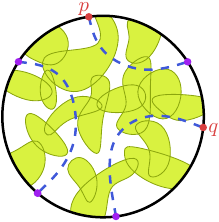} {fpt-intgraph} {fpt-forbidden} {fpt-solution} {fpt-translated2}
{ (a) We may schematically represent $W$ as a circle, since the geometry no longer plays a role. Partial paths are dashed (note that we do not know through which disks these paths will traverse).
  (b) The intersection graph of the regions after adding extra vertices for boundary pieces between points of $S\cup\{p,q\}$, shown in green.
  (c) The secondary graph $H$, representing the forbidden pairs. 
  (d) A possible solution of the vertex multicut problem (highlighted in orange).
  (e) The corresponding cut for the original problem. Once the orange disks have been removed, the endpoints of the partial paths belong to the same region, and thus we can connect them without entering any additional disk (solid paths).}

We consider the intersection graph $G_I$ between the remaining regions of $\regions$. That is, each vertex represents a region of $\regions$, and two vertices are adjacent if and only if their corresponding regions intersect. 
 Similarly to \cite{kla-bcws-05}, we must augment the graph with boundary vertices. The partial paths split the boundary of $R$ into several components. We add a vertex for each component (these vertices are called {\em boundary vertices}). We  connect each such vertex to vertices corresponding to pseudodisks that are adjacent to that piece of boundary (Figure~\ref{fig:fpt-intgraph}). Let $G_{\cal X}=(V_{\cal X}, E_{\cal X})$ be the resulting graph associated to crossing pattern ${\cal X}$. Note that no two boundary vertices are  adjacent.

We now create a secondary graph $H$ as follows: the vertices of $H$ are the boundary vertices of $G_{\cal X}$. We add an edge between two vertices if there is a partial path that separates the vertices in $G_{\cal X}$ (Figure~\ref{fig:fpt-forbidden}). 
Two vertices connected by an edge of $H$ are said to form a {\em forbidden pair} (each partial path that would create the edge is called a {\em witness} partial path). We first give a bound on the number of forbidden pairs that $H$ can have.

\begin{lemma}\label{lem_forbid}
Any crossing pattern has at most $2r^2+r$ forbidden pairs.
\end{lemma}
\begin{proof}
By definition, $G_{\cal X}$ only adds edges between boundary vertices. Thus, it suffices to show that $G_{\cal X}$ has at most $2r+1$ boundary vertices. Since partial paths cannot cross, each such path creates a single cut of the domain. This cut introduces a single additional boundary vertex (except the first partial path that introduces two vertices). Recall that we can map the partial paths to crossings between paths $\pi$ and $\rho$ and, as argued in the proof of Lemma~\ref{lem_patterns}, these paths can cross at most $2r$ times. Thus, we conclude that there cannot be more than $2r+1$ boundary vertices. 
\end{proof}

The following lemma shows the relationship between the vertex multicut problem and the minimum resilience path for a fixed  pattern.

\begin{lemma}\label{lem_key}
There are $k$ vertices of $G_{\cal X}$ whose removal disconnects all forbidden pairs if and only if there are $k$ disks in $\cal D$ whose removal creates a path between $p$ and $q$ that obeys the crossing pattern $\cal X$.
\end{lemma}
\begin{proof}
Consider the regions of $\cal A(D)$ inside $R$
 that are not covered by any disk after the $k$ disks have been removed and let $R'$ be their union. By definition, 
there is a path between $p$ and $q$ with the fixed crossing pattern if all partial paths are feasible (i.e., there exists a path connecting the two endpoints that is totally within $R'$).
The reasoning for each partial path is analogous to the one used by Kumar {\em et al.}~\cite{kla-bcws-05}. If all partial paths are possible, then no forbidden pair can remain connected in $G_{\cal X}$, since---by definition---each forbidden pair disconnects at least one partial path (the witness path).
On the other hand, as soon as one forbidden pair remains connected, there must exist at least one partial path (the witness path) that crosses the forbidden pair. Thus if a forbidden path is not disconnected, there can be no path connecting $p$ and $q$ for that crossing pattern.
\end{proof}

Using Lemma \ref{lem_key}, we can transform the barrier resilience problem to the following one: given two graphs $G=(V,E)$, and  $H=(V,E')$ on the same vertex set, find a set $D \subset V$ of minimum size so that no pair $(u,v)\in E'$ is connected in $G\setminus D$. This problem is known as the (vertex) {\em multicut} problem~\cite{x-sipamc-10}. Although the problem  is known to be NP-hard if $|E'|>2$~\cite{h-mcnf-63}, there exist several FPT algorithms on the size of the cut and on the size of the set $E'$~\cite{m-pgsp-06,x-sipamc-10}. Among them, we distinguish the method of Xiao (\cite{x-sipamc-10}, Theorem 5) that solves the vertex multicut problem in roughly $O((2k)^{k+\ell/2}n^3)$ time, where $k$ is the number of vertices to delete, $\ell=|E'|$, and $n$ is the number of vertices of $G$. 

\begin{theorem}\label{thm:fpt}
Let $\regions$ be a collection of unit disks in $\R^2$, and let $p$ and $q$ be two well-separated points.
There exists an algorithm to test whether $\res(p,q) \leq r $, for any value $r$, and if so, to compute a path with that resilience, in $O(2^{f(r)}n^3)$ time, where $f(r)=r^2\log r+o(r^2\log r)$.
\end{theorem}
\begin{proof}
Recall that our algorithm considers all possible crossings between $\rho$ and $\pi$. For any fixed crossing pattern $\cal X$, our algorithm computes $G_{\cal X}$, and all associated forbidden pairs. We then execute Xiao's FPT algorithm~\cite{x-sipamc-10} for solving the vertex multicut problem. By Lemma~\ref{lem_key}, the number of removed vertices (plus the number of disks that were forced to be deleted by $\cal X$) will give the minimum resilience associated with ${\cal X}$. 

Regarding the running time, the most expensive part of the algorithm is running an instance of the vertex multicut problem for each possible crossing pattern. 
Observe that the parameters $k$ and $\ell$ of the vertex multicut problem are bounded by functions of $r$ as follows: $k \leq r$ and $\ell \leq 2r^2+r$ (the first claim is direct from the definition of resilience, and the second one follows from Lemma~\ref{lem_forbid}). Hence, a single instance of the vertex multicut problem will need $O((2r)^{r+(2r^2+r)/2}n^3)=O(2^{(1+\log r)(r^2+1.5r)}n^3)=O(2^{r^2\log r+o(r^2\log r)}n^3)$ time. By Lemma~\ref{lem_patterns} the number of crossing patterns is bounded by $2^{4r\log r+o(r\log r)}$. Thus, by multiplying both expressions we obtain the bound on the running time, and the theorem is shown. 
\end{proof}

We remark that the importance of this result lies in the fact that an FPT algorithm exists. Hence, although the dependency on $r$ is high, we emphasize that the bounds are rather loose. We also note that both the minimum resilience path and the disks to be deleted can be reported.

\subsection {Extension to Fat Regions}\label{sec:fpt-fat}

\eenplaatje {fat}
{A $\beta$-fat region $D$ is contained in a big disk, but contains a smaller disk; in this example, $\beta=2$.}

We now generalize the algorithm to consider more general shapes. A region $D$ is $\beta$-fat if there exist two concentric disks $C$ and $C'$ whose radii differ by at most a factor $\beta$, such that $C \subseteq D \subseteq C'$ (whenever the constant $\beta$ is not important, the region $D$ is simply called {\em fat}). Figure~\ref {fig:fat} shows an example of a $2$-fat region. However, for our algorithms, it is not sufficient for us to assume that the regions are fat. We impose three restrictions on our fat regions, which make them more like disks: (1) the collection of regions has bounded ply $\ply$, (2) all regions have similar size, allowing us to assume the radius of $C$ is $1$, and the radius of $C'$ is $\beta$, and (3) any two regions have $O(1)$ intersections between their boundaries. Together, these three restrictions ensure that no minimum resilience path traverses a given region more than a constant number of times, making thickness within a constant factor of resilience. 
 We formally describe each restriction, and illustrate how its removal impacts the path complexity.

\tweeplaatjes[scale=1.2] {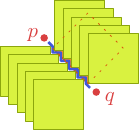}{assumptions-complexity}
{If we eliminate any one of our restrictions we can construct a problem instance whose minimum resilience path must leave and reenter the same (orange) region $\Theta(n)$ times. Here are constructions when removing one of our three restrictions: (a) bounded ply, and (b) bounded region complexity. Note that the case of distinct size was already discussed in Figures~\ref{fig:nounit-size} and~\ref{fig:nounit-ply}. 
}

\begin{description}
\item[Bounded ply]
The arrangement formed by a collection of regions $\regions$ is said to have bounded ply $\ply$ if no point $p\in \R^2$ is contained in more than $\ply$ elements of $\regions$. As we illustrate in Figure~\ref{fig:assumptions-ply}, we can place regions of similar size and bounded region complexity (but no bounded ply) forming a corridor. In particular, the minimum resilience path between $s$ and $t$ may be forced to leave and reenter another similarly-sized region $\Theta(n)$ times. Note that this construction is not possible for unit disks, and therefore unit disk instances do not require bounded ply; however, as soon as we allow a disk with larger radius (e.g., a disk of radius $1 + \epsilon$, $\epsilon > 0$), the bounded ply restriction is required.

\item[Similar size]
We assume without loss of generality that the radius of $C$ is $1$ and the radius of $C'$ is $\beta$; in this case we will call $D$ a \emph {$\beta$-fat unit region}. As previously shown in Figures~\ref{fig:nounit-size} and~\ref{fig:nounit-ply}, with the existence of a single larger region we can create a corridor of $\Theta(n)$ small interlocking regions with constant ply, and partially cover it with a large region to force the optimal resilience path to leave and reenter the large region $\Theta(n)$ times.

\item[Bounded region complexity]
Our final assumption is that the fat regions cannot be too complex. In particular, we assume that any two region boundaries have $O(1)$ pairwise intersections, ensuring that the intersection between any two regions has $O(1)$ connected components. As shown in Figure~\ref{fig:assumptions-complexity}, we can create a corridor with two regions that have $\Theta(n)$ pairwise boundary intersections with a third region, forcing the minimum resilience path to leave and reenter this third region $\Theta(n)$ times. Note that such complex regions can be formed, for example, by taking the union of $\Theta(n)$ circles with radius 1, with centers that are spaced $(\beta - 1)/n$ apart on a line.
\end{description}

Although these restrictions may seem excessive, previous results have made similar assumptions on input regions, and for the same reason we do here: worst-case configurations are possible even with the simplest inputs. For example, to bound the union complexity of fat $(\alpha,\beta)$-covered regions, Efrat~\cite{efrat-05} assumes constant \emph{algebraic complexity}--that region boundaries can be represented by $O(1)$ algebraic polynomials, implying that the region boundaries have at most $O(1)$ pairwise intersections. Whitesides and Zhao~\cite{whitesides-zhao-1990}, when defining \emph{$k$-admissible} curves, impose further restrictions on their (non-fat) regions, requiring the difference of any two regions to be connected, in order to guarantee linear-size union boundary (see also~\cite{aps-sou-08,PachS99} for alternative proofs of this result). Lastly, de Berg~\cite{deberg-2008} assumes constant density, which bounds the number of regions that can intersect any small disk, similar in spirit to ply.

To our knowledge, no definition of fatness meets any of our three assumptions. Fortunately, our assumptions are not overly restrictive. Indeed, they are representative of cases that we are likely to encounter in practice, as it is inefficient to place sensors so that many of them cover the same region, sensor ranges are typically of similar size, and limiting the boundary intersections encompasses both unit disks and pseudodisks as special cases.

The main workings of the algorithm remain unchanged. We start by extending Lemmas~\ref {lem:AtMostTwice}, \ref{lem:dist1.5}, \ref{lem_pathtau}, \ref{lem_patterns} and \ref{lem_forbid} to consider $\beta$-fat unit regions. 

\begin{lemma}
\label{lem:AtMostConstant}
Let $\regions$ be a set of $\beta$-fat unit regions forming an arrangement with ply $\ply$, and bounded region complexity. Let $S \subset \regions$ be an optimal solution. 
In the sequence of regions of $S$ found when going from $p$ to $q$ in an optimal way, no region of $S$ appears more than $O((2\beta+1)^2\ply)$ times.
\end{lemma}
\eenplaatje {fatvisit}{Example showing we can have $\Theta(\beta^2)$ pairwise disjoint $\beta$-fat regions (in orange) that intersect a fixed region (green). By placing a constant number of regions, we can force the minimum resilience path to follow around the boundary of the green region, causing it to enter and leave $\Omega(\beta^2)$ times. This construction has overall constant ply, so we can repeat it until we reach the maximum ply $\Delta$ and get the $\Omega(\ply\beta^2)$ lower bound}

\begin {proof}
  Let $D$ be a region in $S$, and consider its containing disk $C'$ with center $c$. Analogously to the original argument by Bereg and Kirkpatrick~\cite{bk-abrwsn-09}, we note that every time the optimal path visits and leaves $D$, it must do so to avoid some other region. This other region must intersect $D$, and since it is $\beta$-fat unit, it must contain a unit disk centered at distance at most $\beta$ from $D$. 

Therefore all regions intersecting $D$ have their unit-disks centered at distance at most $2\beta$ from $c$. In particular, their unit-disks are totally contained in a disk of radius $2\beta + 1$ centered at $c$. A simple area argument shows that at most $(2\beta+1)^2$ disjoint unit-disks fit into a disk of radius $(2\beta+1)$. Since the ply is bounded by $\ply$, overall there can be up to $\ply(2\beta+1)^2$ regions intersecting $D$. Recall that, by our fatness assumption, two regions can intersect only in $O(1)$ connected components. Therefore, the number of times an optimal path can reenter region $D$ is, proportional to the number of other regions that intersect $D$ which is bounded by $\ply(2\beta+1)^2$. 
\end {proof}

We note that our bound is asymptotically tight. Figure~\ref {fig:fatvisit} illustrates how a matching lower bound. 

\begin{corollary}
\label{cor:AtMostConstant}
When the regions of  $\regions$ are $\beta$-fat unit regions forming an arrangement with ply $\ply$, and bounded region complexity, the thickness between two points is at most $\ply(2\beta+1)^2$ times their resilience.
\end{corollary}

This change in the upper bound of the thickness in terms of the resilience implies similar changes in Lemmas~\ref{lem:dist1.5}, \ref{lem_pathtau}, \ref{lem_patterns} and \ref{lem_forbid}.
The following lemmas summarize these changes; they are proved in the same way as their counterparts for disks, thus we only sketch the differences with the original proofs (if any).

\begin {lemma}
When the regions of  $\regions$ are $\beta$-fat unit regions forming an arrangement with ply $\ply$, and bounded region complexity, the minimum resilience path between $p$ and $q$ cannot traverse cells whose thickness to $p$ or $q$ is larger than $(1+\ply(2\beta+1)^2)\frac{t}{2}$.
\end {lemma}
\begin{proof}
We use the same reasoning as in the proof of Lemma~\ref{lem:dist1.5}. On the one hand there is the minimum thickness path between $p$ and $q$, whose thickness is $t$. 
On the other hand, we also have the minimum resilience path $\rho$ between the same points, whose thickness is at most $\ply(2\beta+1)^2 t$ by Corollary~\ref{cor:AtMostConstant}. Assume now that any cell $C$ traversed by $\rho$ has thickness $k \ply(2\beta+1)^2 t$ from $p$, for some $0<k<1$. The alternative path goes from $p$ to $C$, via $q$, and its thickness is at most $(1-k)\ply(2\beta+1)^2 t + t$. The bound we need is obtained for the value of $k$ that makes both expressions equal, which is  $k=\frac{1}{2}+\frac{1}{2\ply(2\beta+1)^2}$, leading to the claimed value.
\end{proof}

Thus, for $\beta$-fat objects our domain $R$ now becomes be the union of the cells of the arrangement that have thickness from $p$ at most $(1+\ply(2\beta+1)^2)\frac{t}{2}$.

\begin {lemma}
There exists a (possibly non-simple) path $\pi$ whose thickness is at most $(3+\ply(2\beta+1)^2)\frac{t}{2}$, that connects $q$ to a point $q'$ on the outer boundary of $R$ and passes through $p$.
\end {lemma}

\begin{lemma}\label{lem_crosspatbeta}
For any problem instance $\regions$, there are at most $2^{O(\ply^2\beta^4r + \ply\beta^2 r \log ({\ply\beta r}))}$ crossing patterns between $\pi$ and $\rho$. 
\end{lemma}
\begin{proof}
Let $\mu=\ply(2\beta+1)^2$ and  $\nu=\frac{3+\ply(2\beta+1)^2}{2}$. We proceed as in the proof of Lemma~\ref{lem_patterns}. Recall that previously we had $2^{2.5t}\times 2r \cdot (5t \times 5t \times 9)^{2r}$ crossing patterns, but now we must use the bounds that depend on $\beta$ instead. What before was $2r$ now becomes $\mu r$, and the $2.5t$ terms now become $\nu t$. Making these changes in the previous expression, we obtain that the number of crossings is bounded by

$$
2^{\nu t}\times \mu r\times (2 \nu t \times 2\nu t\times 9)^{\mu r}.
$$

Since $t \leq \mu r$ (and by simplifying the expression), this is upper bounded by


$$
2^{\nu \mu r}\times \mu r \times (6\nu \mu r)^{2\mu r}
=
2^{\nu \mu r+\log (\mu r)+2\mu r\log (6\nu \mu r)}.
$$

Finally, we apply that both $\mu,\nu\in O(\ply\beta^2)$, and obtain the desired bound.
%
%
\end{proof}

\begin{lemma}\label{lem_xpattbeta}
Any crossing pattern has at most $O (\ply^2\beta^4 r^2)$ forbidden pairs.
\end{lemma}
\begin{proof}
As in the unit disc case, each crossing between $\pi$ and $\rho$ creates an additional vertex in the boundary (i.e., a potential vertex of $H$). Further note that $\pi$ and $\rho$ can cross at most $2\mu r$ times (since they traverse through at most that many cells of $\arr$). A bound on the number of vertices of $H$ immediately implies a quadratic bound on the number of edges in $H$ as well. Thus, we obtain that the number of forbidden pairs is at most $O((2\mu r+1)^2)=O(\ply^2\beta^4r^2)$ as claimed.
\end{proof}

With these results in place, the rest of the algorithm remains unchanged: the only additional property of unit disks that we use is the fact that they are connected, to be able to phrase the problem as a vertex cut in the region intersection graph.

\begin{theorem}\label{thm:fpt-fat}
Let $\regions$ be a collection of $n$ connected $\beta$-fat unit regions of  bounded region complexity in $\R^2$ forming an arrangement of ply $\ply$, and let $p$ and $q$ be two points. Let $r$ be a parameter.
There exists an algorithm to test whether $\res(p,q) \leq r $, and if so, to compute a path with that resilience, in $O (2^{f(\ply,\beta,r)}n^3)$ time, where $f(\ply,\beta,r)\in O(\ply^2\beta^4r^2\log(\ply \beta r))$.
\end{theorem}

\begin{proof}
As before, the running time is bounded by the product of the number of crossing patterns and the time needed to solve a single instance of the vertex multicut problem. By Lemmas~\ref{lem_crosspatbeta} and~\ref{lem_xpattbeta}, these bounds now become $O(2^{O(\ply^2\beta^4r + \ply\beta^2 r \log ({\ply \beta r}))})$ and $O(2^{O(\ply^2\beta^4r^2\log(\ply \beta r))}n^3)$, respectively. The product of both is dominated by the second term, hence the theorem is shown.
\end{proof}

\section{$(1+\eps)$-approximation} \label {sec:approx}

In this section we present an efficient polynomial-time approximation scheme (EPTAS) for computing the resilience of an arrangement of disks of bounded ply $\ply$. The general idea of the algorithm is very simple: first, we compute all pairs of regions 
that can be reached by removing at most $k$ disks,
for $k = \lceil (16\ply-12) / \eps^2 \rceil$.
 Then, we compute a shortest path in the dual graph of the arrangement of regions, augmented with some extra edges. We prove that the length of the resulting path is a $(1+\eps)$-approximation of the resilience.

As in the previous section, we first consider the case in which $\regions$ is a set of $n$ unit disks in $\R^2$ (note that this time we have the additional constraint that no point is covered in more than $\ply$ disks). 
 Let $\arr$ be the arrangement induced by the regions of $\regions$,
and let $G_{\arr}$ be the dual graph of $\arr$.
Recall that $G_{\arr}$ has a vertex for every cell of $\arr$, and a directed edge between all pairs of adjacent cells of cost $1$ when entering a disk, and cost $0$ when leaving a disk.
For any given $k$, let $G_k$ be the graph obtained from $G_{\arr}$ by adding,
for each pair of cells $A, B \in \arr$ with resilience at most $k$, a \emph{shortcut edge} $\overrightarrow {AB}$ of cost $\res (A,B)$. 

For a pair of cells of $\arr$, we can test whether $r(A,B)$ is smaller than $k$, and if so, compute it, in $O(2^{f(k)}n^3)$ time (where $f(k)=k^2\log k+o(k^2\log k)$) by applying Theorem~\ref {thm:fpt} to a point $p \in A$ and a point $q \in B$. Since the number of pairs of cells of the arrangement is also bounded by a polynomial in $n$, we overall get a EPTAS since $k$ is a constant that depends only on $\eps$ and $\ply$. Again, we emphasize that the bounds presented in this section are not tight, but our objective is to show the existence of an EPTAS for this problem.


\subsection {Analysis}


\begin{lemma}
\label{lem:ResInside}
Let $D \in \regions$, where \arr\ has ply $\ply$, and let $s$,$t$ be any two points inside $D$. 
Then the resilience between $s$ and $t$ in \regions\ is at most $4\ply-3$.
\end{lemma}
\begin{proof}
Let $c$ be the number of disks that contain $s$ or $t$ (or both). Clearly these disks must be removed. Also notice that $c \leq 2\ply-1$, since $D$ contains both points and no point is contained in more than $\ply$ disks. Now we analyze how many other disks may need to be removed too.

Consider a minimum resilience path between $s$ and $t$ among those that stay inside $D$. For each disk $D_1$ (not containing neither $s$ nor $t$) that needs to be removed in an optimal solution, there must be another disk $D_2$ that intersects $D_1$, so that $D_1$ and $D_2$ together separate $s$ and $t$ inside $D$. We call such a pair of disks a \emph{separating pair}. Thus if the resilience is $(c+c')$, there must be at least $c'$ \emph{disjoint}\footnote{By \emph{disjoint} we refer to the identities of the disks, not to the regions they occupy.} separating pairs  intersecting $D$. 
Let $a$ and $b$ be the diametral pair on $D$ that is orthogonal to segment $st$. We claim that one of the disks of any separating pair must cover either $a$ or $b$. Indeed, assume on the contrary that there exists two unit disks $D_1$ and $D_2$ that separate $s$ and $t$ but do not contain neither $a$ nor $b$ (nor $s$ or $t$). Without loss of generality, we may assume that both $s$ and $t$ lie on the boundary of $D$. Observe that in order to separate $s$ and $t$, the union of $D_1$ and $D_2$ must cross segment $st$ and cannot cross segment $ab$ (otherwise it would contain $s$ or $t$, since $D$, $D_1$ and $D_2$ are unit disks). However, the only possible way of doing so is if $D_1$ and $D_2$ are tangent to $a$, $b$ and either $s$ or $t$ (see Figure~\ref{fig:approx-resply3}). However, in this case $s$ and $t$ are not separated, a contradiction.

That is, for each separating pair we have a unique disk that covers either $a$ or $b$. Since no point is contained in more than $\ply$ disks (and $D$ contains both $a$ and $b$ we conclude that there cannot be more than $2(\ply-1)$ separating pairs, completing the proof of the lemma. 

\end{proof}
\eenplaatje [scale=0.8] {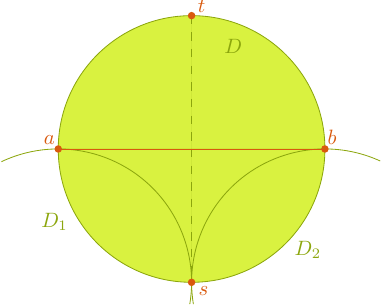}{Barring symmetric configurations, this is the only way of making two disks that cross the segment $st$ (dashed line) and avoids the segment $ab$ (solid line). However, in this case the disks $D_1$ and $D_2$ do not separate $s$ and $t$.}


The previous lemma implies that in an optimal resilience path, if a disk appears twice, the two entry points have resilience at most $4\ply-3$ apart (when counting the cells traversed by the path between the two occurrences of the disk). Note that a lower bound of $\ply$ is also easy to construct, so the result is (asymptotically speaking) tight.

To prove the result in this section it will be convenient to focus on the sequence of disks encountered by a path when going from $p$ to $q$.
It turns out that such problem is essentially a string problem, where each symbol represents a disk encountered by the path. 
In that context, the thickness will be equivalent to the number of symbols of the string (recall that we assume that $p$ is not contained in any disk), and the resilience to the number of distinct symbols.




Let $S = \langle s_1 \ldots s_n \rangle$ be a string of $n$ symbols from some alphabet $\alp$, such that no symbol appears more than twice.
Let $T$ be a substring of $S$.
We define $\ell (T)$ to be the length of $T$, and $d (T)$ to be the number of distinct symbols in $T$. Clearly, $\frac12 \ell (T) \le d (T) \le \ell (T)$.
Let $\sigma$ and $k$ be two fixed integers such that $\sigma < k$.
We define the \emph {cost} of a substring $T$ of $S$ to be:
\[ 
\psi(T) = 
\begin{cases}
  \sigma  & \text{ if } T = \langle \sym a\tau \sym a \rangle \text{ for some } \sym a \in \alp \text{, string } \tau \text{ s.t. } \sym a\not\in\tau, \text{ and }\ell(T) > \sigma, \\
  d(T)    & \text{ if } \ell(T) \le k \textrm{ (and the first condition fails), }\\
  \ell(T) & \text{ otherwise.} \\
\end{cases}
\]

Note that, in the string context, $d$ acts as the resilience, $\ell$ as the thickness, and $\psi$ is the approximation we compute. Intuitively, if $T$ is short (i.e., length at most $k$) we can compute the exact value $d(T)$. If $T$ has a symbol whose two appearances are far away we will use a ``shortcut'' and pay $\sigma$ (i.e., for unit disk regions, by Lemma~\ref{lem:ResInside}, we have $\sigma=4\ply-3$). Otherwise, we will approximate $d$ by $\ell$. 

Given a long string, we wish to subdivide $S$ into a \emph{segmentation} $\cal T$, composed of $m$ disjoint segments (i.e., substrings of $S$) $T_1, \ldots, T_m$, that minimize the total cost $\psi(\mathcal T) = \sum_i \psi(T_i)$.
Clearly, $\psi(\mathcal T) \le \ell(S)$.

\begin {lemma}
\label{lem:segmentation}
  Let $S$ be a sequence.
  There exists a segmentation $\mathcal T$ such that 
  $\psi (\mathcal T) \le (1 + \eps) d(S)$, where
  $\eps = 2\sqrt {\sigma / k}$.
\end {lemma}
\begin {proof}
  Let $\lambda$ be an integer such that $\sigma < \lambda < k$, of exact value to be specified later.
  First, we consider all pairs of equal symbols in $S$ that are more than $\lambda$ apart.
  We would like to take all of these pairs as separate segments;
  however, we cannot take segments that are not disjoint.
%
  So, we greedily take the leftmost symbol $\sym s$ whose partner is more than $\lambda$ further to the right, and mark this as a segment. We recurse on the substring remaining to the right of the rightmost $\sym s$.\footnote
  {In fact, we could choose any disjoint collection such that after their removal there are no more segments of this type longer than $\lambda$.}
  Finally, we segment the remaining pieces greedily into pieces of length $k$.
  Figure~\ref{fig:strings} illustrates the resulting segmentation.
  

  Now, we prove that the resulting segmentation has a cost of at most $(1+\eps) d(S)$.
  First, consider a symbol to be \emph {counted} if it appears in only one short (blue) segment, and to be \emph {double-counted} if it appears in two different short segments.
  Suppose $\sym s$ is double-counted. Then the distance between its two occurrences must be smaller than $\lambda$, otherwise it would have formed a long (red) segment. Therefore, it must appear in two adjacent short segments. The leftmost of these two segments has length exactly $k$, but only $\lambda$ of these can have a partner in the next segment. So, at most a fraction $\lambda / k$ symbols are double-counted.
  
  Second, we need to analyze the cost of the long (red) segments. In the worst case, all symbols in the segment also appear in another place, where they were already counted. In this case, the true cost would be $0$, and we pay $\sigma$ too much. However, we can assign this cost to the at least $\lambda$ symbols in the segment; since each symbol appears only twice they can be charged at most once. So, we charge at most $\sigma / \lambda$ to each symbol. 
  The total cost is then bounded by $(1 + \lambda / k + \sigma / \lambda) d(S)$.
  To optimize the approximation factor, we choose $\lambda$ such that $\lambda / k = \sigma / \lambda$; more precisely we take $\lambda =\lceil \sqrt {k\sigma} \rceil$. 
\end {proof}

Recall that for our resilience approximation we have $\sigma=4\ply-3$ (Lemma~\ref{lem:ResInside}). Thus, the actual value of $k$ is obtained by solving $\eps = 2 \sqrt{\sigma / k}$ for $k$, which leads to  $k = \lceil (16\ply-12) / \eps^2 \rceil$.

\begin{figure}[tb]
  \centering
  \subfigure {\put(-16,2){(a)} \includegraphics[width=\textwidth]{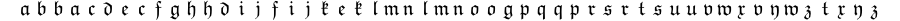}}
\vspace{-15pt}
  \subfigure {\put(-16,6){(b)} \includegraphics[width=\textwidth]{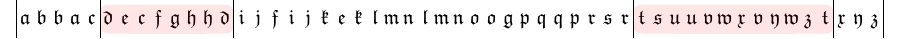}}

  \subfigure {\put(-16,6){(c)} \includegraphics[width=\textwidth]{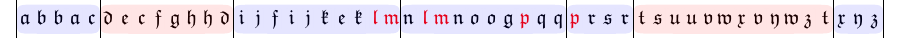}}
\vspace{-15pt}
  \caption 
  { (a) A string of $52$ symbols, each appearing twice.
    (b) First, we identify a maximal set of segments bounded by equal symbols, and longer than $\lambda = 4$.
    (c) Then, we segment the remaining pieces into segments of length $k = 10$. Red symbols are double-counted.
  }  
  \label{fig:strings}
\end{figure}

\subsection {Application to resilience approximation}

We now show that the shortest path between any $p,q$ in $G_k$ is a $(1+\eps)$-approximation of their resilience.
Let $\pi$ be a path from $p$ to $q$ in $\R^2$, and let $S(\pi)$ be the sequence that records every disk of $\mathcal D$ we enter along $\pi$, plus the disks that contain the start point of $\pi$,  added at the beginning of the sequence, in any order.
Then we have $|S(\pi)| = \thk(\pi)$.

\begin {lemma} \label{lem:alg_segmentation}
For every path $\pi$ from $p$ to $q$ and every segmentation $\cal T$ of $S(\pi)$, there exists a path from $p$ to $q$ in $G_k$ of cost at most $\psi (\cal T)$.
\end {lemma}
\begin {proof}
  We describe how to construct a path in $G_k$ based on $\cal T$.
For every segment $T$ of $\cal T$,
  we create a piece of path whose length in $G_k$ is at most the cost of the segment $\psi(T)$.
  
 There are three types of segments.
 The first type are segments that start and end with the same symbol $\sym a$, which corresponds to a disk $D \in \cal D$. For those, we make a shortcut path that stays inside $D$, as per Lemma~\ref {lem:ResInside}.
 The second type are segments whose length is at most $k$. For those, by definition, $G_k$ contains a shortcut edge whose cost is exactly the resilience between the corresponding cells of $\arr$.
 The third type are the remaining segments. For those, we simply use the piece of $\pi$ that corresponds to $T$.
%
%
\end {proof}

\begin{figure}[tb]
  \centering
  \subfigure {{(a)} \includegraphics[width=0.40 \textwidth]{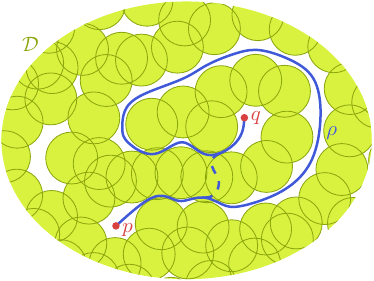}}
  \subfigure {{(b)} \includegraphics[width=0.40 \textwidth]{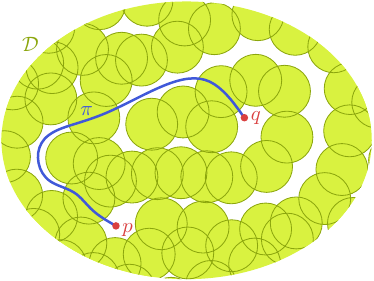}}
  \caption 
  { (a) The optimal path $\rho$, achieving a resilience of $2$.
      There is a segmentation of $\rho$ of cost $3$, using the dashed shortcut.
    (b) A minimum cost path $\pi$ found by the algorithm. 
      In this example, the resilience of $\pi$ is $3$.
  }  
  \label{fig:approx-optimal}
\end{figure}
\begin {lemma}\label{lem:apx-res}
  For any $p,q\in\R^2$, it holds $\cost_{G_k}(\shpa_{G_k}(p,q))\leq (1+\eps)\res(p,q)$.
\end {lemma}
\begin {proof}
  Let $\rho$ be a path from $p$ to $q$ of optimal resilience $r^* = \res(\rho) = \res(p,q)$.
  Then, consider the sequence $S(\rho)$, that is, the sequence of disks that $\rho$ enters.
  Now, by Lemma~\ref {lem:segmentation}, there exists a segmentation $\cal T$ of $S(\rho)$ of cost at most $(1 + \eps)d(S(\rho)) = (1 + \eps)r^*$.
  By Lemma~\ref {lem:alg_segmentation}, there exists a path in $G_k$ of equal or smaller cost.
  Figure~\ref {fig:approx-optimal} illustrates this.

  Now, consider the path $\pi$ that our algorithm produces. The resilience of $\pi$ is smaller than the cost of $\pi$ in $G_k$, which is smaller than the cost of $\rho$ in $G_k$, which is smaller than $1+\eps$ times the resilience of $\rho$.
  That is:
  $    r(\pi) \le \cost_{G_k}(\pi) \le \cost_{G_k}(\rho) \le (1+\eps)r(\rho) = (1+\eps)r^*
  $.
\end {proof}

\begin {theorem}\label {thm:approx}
  Let $\regions$ be a set of unit disks of ply $\ply$ in $\R^2$. We can compute a path $\pi$ between any two given points $p,q\in\R^2$ whose resilience is at most $(1+\eps)r(p,q)$ in $O(2^{f(\ply, \eps)}n^5)$ time, where $f(\ply,\eps) = O\left(\frac{\ply^2\log(\ply / \eps)}{\eps^4}+o\left(\frac{\ply^2\log(\ply / \eps)}{\eps^4}\right)\right)$.
\end {theorem}
\begin{proof}
The running time of the algorithm is dominated by the preprocessing stage: determining if the resilience between every pair of vertices of $G_{\arr}$ is at most $\lceil (16\ply-12) / \eps^2 \rceil$. Since $G_{\arr}$ is an arrangement of disks with ply at most $\ply$, it has $O(\ply n)$ cells\footnote{We thank the anonymous referee that pointed this to us and allowed the dependency in $n$ to be lowered.}. We execute the algorithm of Theorem~\ref{thm:fpt} for every pair of cells (thus, $O(\ply^2 n^2)$ times), and we obtain the desired bound.
\end{proof}

\subsection {Extension to fat regions} \label {sec:approx-fat}

As in Section~\ref {sec:fpt-fat}, we now generalize the result to arbitrary $\beta$-fat unit regions. We again assume that our collection of regions has bounded ply $\ply$, and that the region boundaries have $O(1)$ pairwise intersections. As in Section~\ref{sec:fpt-fat}, for simplicity in the notation our analysis assumes that the region boundaries have at most two pairwise intersections, implying that the intersection between any two overlapping regions has one connected component. However, our results generalize to $k = O(1)$ pairwise intersections between region boundaries.

\begin{lemma}
\label{lem:FatResInside}
Let $D \in \regions$, where \arr\ has ply $\ply$, and let $p$,$q$ be any two points inside $D$. 
Then the resilience between $p$ and $q$ in \regions\ is at most $(2\beta+1)^2\ply$.
\end{lemma}
\begin{proof}
The resilience between $p$ and $q$ is upper-bounded by the number of regions that intersect $D$. 
We can give an upper bound using a simple packing argument.
Since $p$ and $q$ belong to a $\beta$-(unit)fat region $D$, they are both inside a circle $C$ with center $c$ and radius $\beta$.
Any other $\beta$-fat region $D'$ that interferes with the path from $p$ to $q$ must intersect $C$. 
Such an intersecting region, being also $\beta$-fat, must contain a unit-disk whose center cannot be more than $2\beta$ away from $c$.
Therefore all regions intersecting $C$ have their unit-disks centered at distance at most $\beta$ from $c$.
Moreover, such disks are totally contained in a disk of radius $2\beta + 1$ centered at $c$.
As in the proof of Lemma~\ref{lem:AtMostConstant}, we can show that at most $(2\beta+1)^2$ disjoint unit-disks fit into a disk of radius $(2\beta+1)$.
Since the ply is at most $\ply$, the maximum number of unit-disks inside a disk of radius $\beta$ in $\regions$ is $(2\beta+1)^2 \ply$.
\end{proof}

As before, the rest of the arguments do not rely on the geometry of the regions anymore, and we can proceed as in the disk case. The only difference is that the value $\sigma$ of doing a shortcut has increased to $(2\beta+1)^2\ply$. 

\begin {theorem}\label {thm:approx-fat}
  Let $\regions$ be a set of $\beta$-fat regions of ply $\ply$ in $\R^2$. We can compute a path $\pi$ between any two points $p,q\in\R^2$ whose resilience is at most $(1+\eps)r(p,q)$ in $O(2^{f(\ply,\beta,\eps)}n^5)$ time, where $f(\ply,\beta,\eps) = O\left(\frac{\ply^4\beta^8}{\eps^4}\log(\beta\ply/\eps)\right)$.
\end {theorem}

\section{NP-hardness} \label {sec:np}
In this section we show that computing the resilience of certain types of fat regions is NP-hard. We recall that several NP-hardness results for other shapes are already known, but most of them are for skinny objects. For example, hardness for the case in which regions are line segments in $\R^2$ was shown in~\cite{acgk-cpslsa-11,tk-brsn-11} and ~\cite[Section 5.1]{y-sppacgatm-12}. Our contribution is to show that hardness holds for for the case in which ranges have bounded fatness (i.e., ranges are not skinny). The only hardness proof that we know for objects of positive area is by Tseng~\cite{t-rwsn-11}, who shows that if the regions are rotations and translations of a fixed square or ellipse the problem is NP-hard. 

In addition to showing that the problem is difficult for other shapes, our construction is of independent interest, since it is completely different from those given in~\cite{acgk-cpslsa-11},~\cite{t-rwsn-11},~\cite {tk-brsn-11}, and~\cite[Section 5.1]{y-sppacgatm-12}. Moreover, our proof has the advantage of being easy to extend to other shapes. We also note that the construction of Tseng uses several rotations of a fixed shape (i.e., 3 for a square, 4 for an ellipse), whereas our construction only needs two different rotations of the same shape.

First we show NP-hardness for general connected regions, and later we extend it to axis-aligned rectangles of aspect ratio $1:1+\eps$ and $1+\eps:1$. We start the section establishing some useful graph-theoretical results.


  Let $G$ be a graph, and let $p$ be a point in the plane. Let $\Gamma$ be an embedding of $G$ into the plane, which behaves properly (vertices go to distinct points, edges are curves that do not meet vertices other than their endpoints and do not triple cross), and such that $p$ is not on a vertex or edge of the embedding. We say $\Gamma$ is an \emph {odd} embedding around $p$ if it has the following property: every cycle of $G$ has odd length if and only if the winding number of the corresponding closed curve in the plane in $\Gamma$ around $p$ is odd.
%
 We say a graph $G$ is \emph {oddly embeddable} if there exists an odd embedding $\Gamma$ for it (Figure~\ref {fig:odd-planar+odd-nonplanar+odd-tripartite} shows some examples). We claim that vertex cover is NP-hard for this constrained class of graphs. The proof of this statement  is based on two observations.
%



%
%


  \drieplaatjes [scale=0.82] {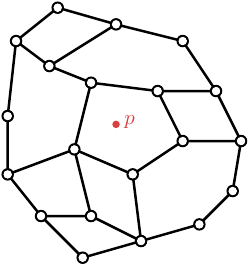} {odd-nonplanar} {odd-tripartite}
  { (a) A planar odd embedded graph. 
    (b) A non-planar one.
    (c) A tripartite graph, oddly embedded around $p$.
  }

  \begin {observation}
    Every tripartite graph is oddly embeddable.
  \end {observation}
  \begin {proof}
    The vertices of a tripartite graph $G$ can be divided into three groups $V_1, V_2, V_3$ such that there are no internal edges in any of these groups. Now, consider a triangle $\ply$ around $p$. We create an embedding $\Gamma$ where all vertices in $V_1$ are close to one corner of $\ply$, the vertices in $V_2$ are close to a second corner, and the vertices in $V_3$ are close to the remaining corner. All edges are straight line segments.
    See Figure~\ref {fig:odd-tripartite}.    
    
    Consider the graph $H$ obtained from $G$ by contracting all vertices in $V_i$ to a single vertex $v_i$; $H$ is a triangle (or a subgraph of a triangle).
    Now consider any cycle in $G$, and project it to $H$. Since there were no edges in $G$ connecting vertices within a group $V_i$, this does not change the length of the cycle, nor does it change the winding number around $p$. 
    Any two consecutive edges from $v_i$ to $v_j$, and back from $v_j$ to $v_i$, do not influence the parity of the length of the cycle, nor the winding number around $p$, so we can remove them from the cycle. We are left with a cycle of length $3w$ and winding number $w$ or $-w$, for some integer $w$. Clearly, $3w$ is odd if and only if $w$ is odd.
    Therefore, $\Gamma$ is an odd embedding of $G$, as required.
  \end {proof}


  The \emph {maximum independent set} problem in a graph asks for the largest set of vertices in the graph such that no two vertices in the set are connected by an edge. This problem is well-known to be NP-hard on general graphs. In fact, it remains NP-hard for tripartite graphs. A simple proof is included for completeness, and because we need the argument later. Note that a minimum vertex cover is the complement of a maximum independent set, hence by proving the NP-hardness of maximum independent set, we are also proving that minimum vertex cover is NP hard.

  \begin {observation} \label{obs:insert-even}
    Let $G=(V,E)$ be a graph.
    Let $G'$ be obtained from $G$ by subdividing every edge $e \in E$ into an odd number of pieces, by adding an even number $m_e$ of new vertices. Let $m = \sum_{e} m_e$ be the total number of vertices added.
    Then $G$ has a maximum independent set $I \subset V$ if and only if $G'$ has a maximum independent set $I'$ with $|I'| = |I| + m/2$.
  \end {observation}
  
  \begin {proof}
    For every independent set $I \subseteq V$ in $G$, there is a corresponding independent set $I'$ in $G'$ with $|I'| = |I| + m/2$: for every pair of extra vertices on an edge, we can always add one of the two to an independent set.
    Conversely, for every independent set $I'$ in $G'$, there is a corresponding independent set $I$ in $G$ with $|I| = |I'| - m/2$: $I'$ cannot use both extra vertices on an edge, so if we simply remove all extra vertices we remove at most $|E|$ elements from $I'$ (clearly, if we remove less than $m/2$ vertices from $I'$ this way, we can remove more vertices until $I$ has the desired cardinality).
  \end {proof}

From the above observations, it follows that maximum independent set is also NP-hard on tripartite graphs,  and hence, also on oddly embeddable graphs. 
Since our construction does not increase the maximum vertex degree, and vertex cover is known to be NP-hard for graphs with maximum degree three, we obtain the following.

  \begin{corollary}\label{cor_np}
    Minimum vertex cover on oddly embeddable graphs of maximum degree 3 is NP-hard.
  \end {corollary}

  Given an embedded graph $\Gamma$, we say that a curve in the plane is an \emph {odd Euler path} if it does not go through any vertex of $\Gamma$ and it crosses every edge of $\Gamma$ an odd number of times. 
  
  \begin {lemma} \label {lem:eulerpath}
    Let $p$ be a point in the plane, and $\Gamma$ an oddly embedded graph around $p$. Then there exists an odd Euler path for $\Gamma$ that starts at $p$ and ends in the outer face. Moreover, such path can be computed in polynomial time.
  \end {lemma}
  \drieplaatjes [scale=0.65] {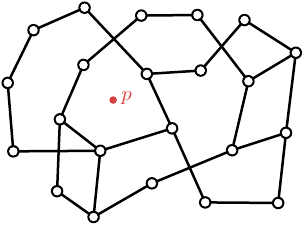} {crossing-flatten} {crossing-tour} {(a) An oddly embedded graph with four crossings. (b) The crossings are flattened according to the parity of their vertices. (c) An odd Euler path from $p$ to the outer face.}
  
  \begin {proof}
    First, we insert an even number of extra vertices on every edge of $\Gamma$ such that in the resulting embedded graph $\Gamma'$, every edge crosses at most one other edge. Now we construct an Euler path that crosses every edge of $\Gamma'$ exactly once; note that this path will therefore cross every edge of $\Gamma$ an odd number of times.
    Consider a pair of crossing edges and the four vertices concerned. For each pair of consecutive vertices (vertices that are not endpoints of the same edge), find a path in the graph that does not go around $p$ (when seen as a cycle, after adding the crossing).

     The parity of the length of this path does not depend on which path we take: if there would be an even-length path and an odd-length path between the same two vertices, both of which do not go around $p$, then there would be an odd cycle that does not contain $p$, which contradicts the oddly embeddedness of $\Gamma'$. Now, if the path has even length, we identify these two vertices. 
Note that of the four pairs of vertices involved in a crossing (i.e., ignoring the two pairs forming edges in $\Gamma$), exactly two pairs will have odd length connecting paths, so effectively we ``flatten'' the crossing. We do this for all crossings, and call the resulting multigraph $\Gamma''$. See Figure~\ref {fig:crossing-flatten}.
(If the two crossing edges belong to different connected components of $\Gamma$, there are no paths connecting their vertices; in this case we make an arbitrary choice of which vertices to identify.)
    
    Now $\Gamma''$ is planar. Furthermore, by construction, all faces of $\Gamma''$ have even length, except the one containing $p$ and the outer face. Therefore, the dual multigraph of $\Gamma''$ has only two vertices of odd degree, and hence has an Euler path between these vertices. Furthermore, this Euler path crosses every edge of $\Gamma'$ exactly once, and therefore every edge of $\Gamma$ an odd number of times. Note that the proof is constructive. Moreover, both the transformations and the Euler path can be done in polynomial time, hence such path can also be obtained in polynomial time. 
  \end {proof}



  
  \begin {lemma} \label {lem:odd-independent-general}
    Let $p$ be a given point in the plane, and $\Gamma$ an oddly embedded graph (not necessarily planar) around $p$. Furthermore, let $T$ be a curve that forms an odd Euler path from $p$ to the outer face. Then we can construct a set $\mathcal D$ of connected regions such that a minimum set of regions from $\mathcal D$ to remove corresponds exactly to a minimum vertex cover in $\Gamma$.
  \end {lemma}
  
    \tweeplaatjes [scale=0.8]{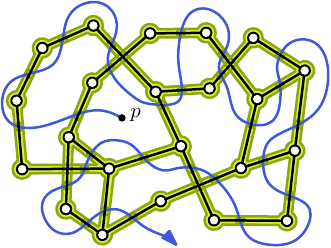} {general-wall} {Creating regions to follow $\Gamma$ and $T$.}
  
  \begin {proof}
If $T$ is self-intersecting, then we can rearrange the pieces between self-intersections to remove all self-intersections. Thus we assume that $T$ is a simple path. 
    
If $T$ crosses any edge of $\Gamma$ more than once, we insert an even number of extra vertices on that edge such that afterwards, every edge is crossed exactly once. 
 Let $\Gamma'$ be the resulting graph. Since we inserted an even number of vertices on every edge, finding a minimum vertex cover in $\Gamma'$ will give us a minimum vertex cover in $\Gamma$.
    
    Now, for each vertex $v$ in $\Gamma'$, we create one region $D_v$ in $\mathcal D$. This region consists of the point where $v$ is embedded, and the pieces of the edges adjacent to $v$ up to the point where they cross $T$.
    Figure~\ref {fig:general-regions} shows an example (the regions have been dilated by a small amount for visibility; if the embedding $\Gamma$ has enough room this does not interfere with the construction). Note that all regions are simply connected.
    
    Finally, we create one more special region $W$ in $\mathcal D$ that forms a corridor for $T$. Then $W$ is duplicated at least $n$ times to ensure that crossing this ``wall'' will always be more expensive than any other solution. Figure~\ref {fig:general-wall} shows this.
    
    Now, in order to escape, anyone starting at $p$ must roughly follow $T$ in order to not cross the wall. This means that for every edge of $\Gamma'$ that $T$ passes, one of the regions blocking the path (one of the vertices incident to the edge) must be disabled. The smallest number of regions to disable to achieve this corresponds to a minimum vertex cover in $\Gamma'$.
  \end {proof}
  
  Combining this result with Corollary~\ref{cor_np}, we obtain our first hardness result for the barrier resilience problem. 

\begin{theorem}
The barrier resilience problem for a collection of connected regions is NP-hard.
\end{theorem}
  
    
\subsection{Extension to fat regions} \label {app:hardfat}
  We now adapt the previous approach to also work for a much more restricted class of regions:
  axis-aligned rectangles of sizes $1 \times (1+\eps)$ and $(1+\eps) \times 1$ for any $\eps > 0$ (as long as $\eps$ depends polynomially on $n$).
For simplicity, we limit $\Gamma$ to have maximum degree 3. 
Maximum independent set is still known to be NP-hard in that case~\cite{gjs-ssnpcgp-76}, and making them tripartite does not change the maximum degree.

The idea of the reduction is the following.
We start from a sufficiently spacious (but polynomial) embedding of $\Gamma$,
as illustrated in Figure~\ref {fig:total-grid}.
On each edge we add a large even number of extra vertices. 
Each new vertex will be replaced by a rectangle, so every edge in $\Gamma$ will become a chain of overlapping rectangles, like the green rectangles in Figure~\ref {fig:total-graphrects}.
Therefore the first phase consists in replacing the embedding of $\Gamma$ by an equivalent embedding of rectangles. We call these rectangles \emph {graph rectangles} (green in the figures).
Some care must be taken in the placement of graph rectangles around degree-3 vertices and in crossings, so that
the rest of the construction can be made to work.
Next, we place \emph {wall rectangles} (orange in the figures; these consist of many copies of the same rectangle) across each graph rectangle. The gaps between adjacent wall rectangles should cover the overlapping part of two adjacent graph rectangles, so that a path can pass through them only whenever one of the two graph rectangles is removed.
Then, we find a curve $T$ from $p$ that goes through every gap exactly once (note that $T$ exists, by Lemma~\ref {lem:eulerpath}).
Figure~\ref {fig:total-path} illustrates this phase of the construction.
Finally, we add more wall rectangles around $T$, to force any potential minimum resilience path from $p$ that does not go through the wall rectangles to be homotopic to $T$.
Figure~\ref {fig:total-final} shows the final set of rectangles.
Now, computing an optimal resilience path among this set of rectangles would correspond to a maximal independent set in $\Gamma$.

\vierplaatjes [scale=.29] 
{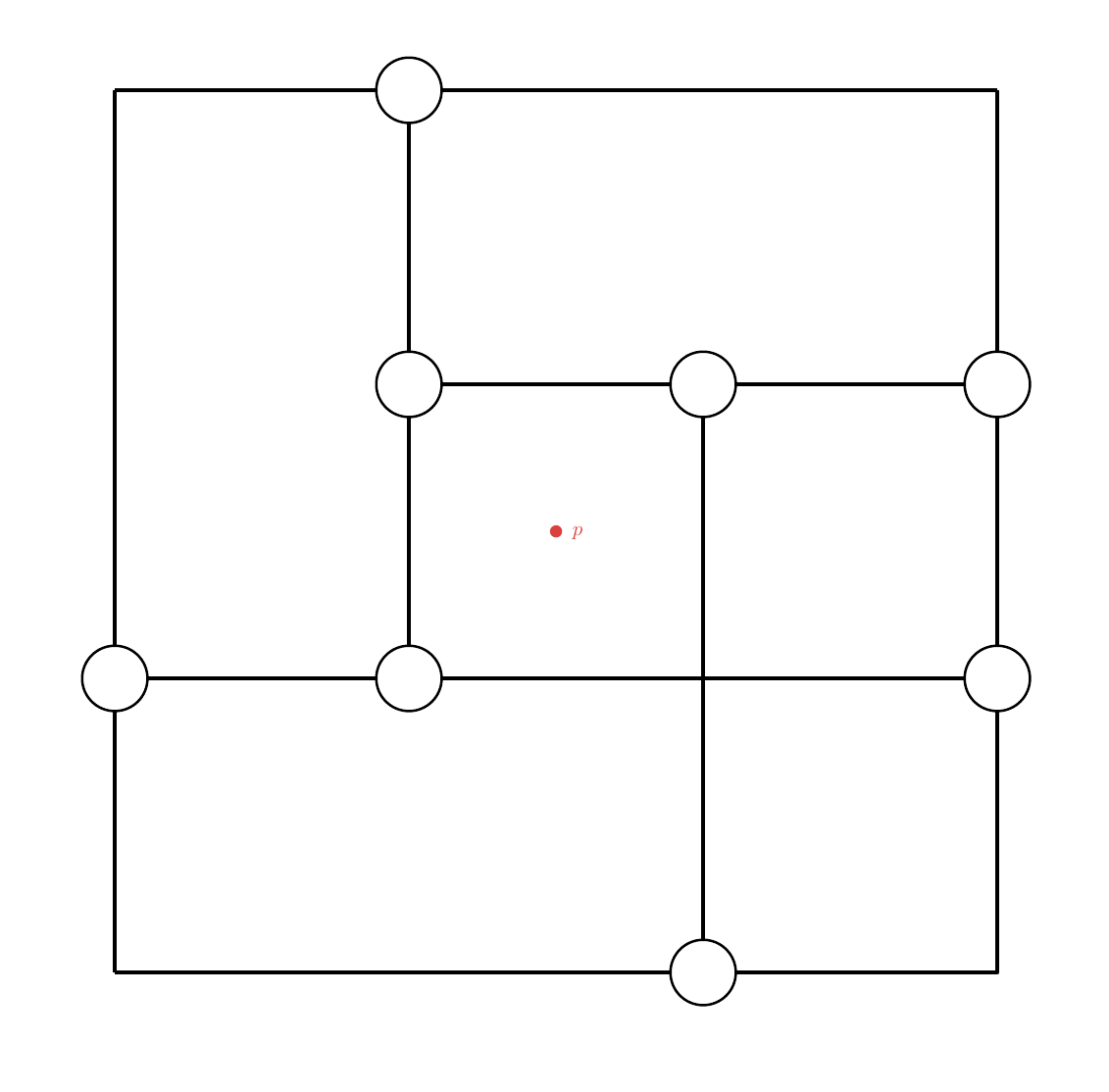} 
{total-graphrects} 
{total-path} 
{total-final} 
{(a) A non-planar oddly embedded cubic graph, embedded on a grid.
 (b) A set of rectangles, containing exactly one rectangle for each input vertex, and an even number of rectangles for each input edge. Note that crossings can be embedded if sufficiently far apart.
 (c) Local walls are added to make ``tunnels'', each tunnel contains the overlapping part of two adjacent yellow rectangles. To go through a tunnel, one of the two yellow rectangles has to be removed.
     Then we choose an Euler path from $p$ to the outside, that goes through each tunnel exactly once.
 (d) The final set of rectangles, designed to force any path from $p$ to the outside to be homotopically equal to the one we drew.
}

  \drieplaatjes [scale=.82] {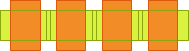} {rects-vertex} {rects-crossing} 
  { Local details of the construction. Note that we use rectangles with a large aspect ratio for visibility,
    but the same constructions can be made with aspect ratio arbitrarily close to $1$.  
    (a) Overlapping rectangles to create edges (with an even number of extra vertices).
    (b) A vertex of degree at most $3$, which is just a single rectangle. 
    (c) A crossing between two chains.
  }
  
For the construction to work, there needs to be enough space to place the wall rectangles. It is clear that this is possible far away from the graph rectangles, but close to the graph rectangles we proceed as follows: first, Figure~\ref {fig:rects-chainsimple} shows the placement of rectangles along an edge of $\Gamma$.
Figure~\ref {fig:rects-vertex} shows how to place the rectangles at degree-3 vertices.
Crossings are handled as shown in Figure~\ref {fig:rects-crossing}. 
These gadgets force some of the gaps in the chain to join each other. But this is no problem if every edge has enough rectangles. Also, note that at the center of the construction in Figure~\ref {fig:rects-crossing} there are two overlapping green rectangles, which belong to the two crossing chains. This is the only place where we vitally use the fact that the regions are not pseudodisks.

  \begin {lemma}
    Let $p$ be a given point in the plane, and $\Gamma$ an oddly embedded graph with maximum vertex degree $3$ (not necessarily planar) around $p$. Furthermore, let $T$ be a curve that forms an odd Euler path from $p$ to infinity. Then we can construct a set $\mathcal D$ of axis-aligned rectangles of aspect ratio $1 : (1 + \eps)$ such that a minimum set of regions from $\mathcal D$ to remove corresponds exactly to a minimum vertex cover in $\Gamma$.
  \end {lemma}
    
  \begin {proof}
    We first add groups of extra vertices on every edge of $\Gamma$ so that we have room to place the rectangles, in an even number per edge. Then replace edges by chains as of rectangles as as in Figure~\ref {fig:rects-chainsimple+rects-vertex+rects-crossing}, and connect the orange (wall) rectangles to force the only optimal path from $p$ to the outer face to be along the Euler path $T$. 
The path may have to be rerouted locally close to the crossings, but since there is a sufficiently large number of crossings with every edge anyway, this is always possible.
 Orange rectangles have to be duplicated sufficiently many times again, to make sure that no optimal path will ever cross them.
  \end {proof}

  
\begin{theorem}
The barrier resilience problem for regions that are axis-aligned rectangles of aspect ratio $1 : (1 + \eps)$ is NP-hard.
\end{theorem}

  A similar approach can likely be used to show NP-hardness of other classes of regions as well. However, it seems that a necessary property for our approach is that the regions are able to completely cross each other: in other words, the regions in $\mathcal D$ cannot be pseudodisks.\footnote{A similar fact was also observed in~\cite{tk-brsn-11}.} 
  

\paragraph{Acknowledgments}
The authors would like to thank some anonymous referees for their thorough check of a previous version of this document. 
M.K was  partially supported by the ELC project (MEXT KAKENHI No. 24106008). M.L. was supported by the Netherlands Organisation for Scientific Research (NWO) under grant 639.021.123.
R.I.~S. was partially supported by projects MINECO MTM2015-63791-R/FEDER and Gen. Cat. DGR 2014SGR46, and by MINECO through the Ram{\'o}n y Cajal program.

\bibliographystyle{abbrv} 
\bibliography{resilience}

\begin{thebibliography}{10}

\bibitem{aps-sou-08}
P.~K. Agarwal, J.~Pach, and M.~Sharir.
\newblock {\em Surveys on Discrete and Computational Geometry: Twenty Years
  Later}, volume 453 of {\em Contemporary Mathematics}, chapter State of the
  Union (of Geometric Objects).
\newblock AMS, 2008.

\bibitem{acgk-cpslsa-11}
H.~Alt, S.~Cabello, P.~Giannopoulos, and C.~Knauer.
\newblock On some connection problems in straight-line segment arrangements.
\newblock In {\em Proc. EuroCG}, pages 27--30, 2011.
\newblock Also available as {CoRR abs/1104.4618}.

\bibitem{bk-abrwsn-09}
S.~Bereg and D.~G. Kirkpatrick.
\newblock Approximating barrier resilience in wireless sensor networks.
\newblock In {\em Proc. ALGOSENSORS}, pages 29--40, 2009.

\bibitem{Cabello2016}
S.~Cabello and P.~Giannopoulos.
\newblock The complexity of separating points in the plane.
\newblock {\em Algorithmica}, 74(2):643--663, 2016.

\bibitem{ck-mpamcppaabr-14}
D.~Y.~C. Chan and D.~G. Kirkpatrick.
\newblock Multi-path algorithms for minimum-colour path problems with
  applications to approximating barrier resilience.
\newblock {\em Theoretical Computer Science}, 553:74--90, 2014.

\bibitem{c-kbcm-12}
C.-Y. Chang, C.-Y. Hsiao, and C.-T. Chang.
\newblock The k-barrier coverage mechanism in wireless visual sensor networks.
\newblock In {\em Proc. {IEEE} WCNC}, pages 2318--2322, 2012.

\bibitem{cglw-ammsm-12}
D.~Z. Chen, Y.~Gu, J.~Li, and H.~Wang.
\newblock Algorithms on minimizing the maximum sensor movement for barrier
  coverage of a linear domain.
\newblock {\em Discrete {\&} Computational Geometry}, 50(2):374--408, 2013.

\bibitem{deberg-2008}
M.~de~Berg.
\newblock Improved bounds on the union complexity of fat objects.
\newblock {\em Discrete Comput. Geom.}, 40(1):127--140, July 2008.

\bibitem{bcko-aad-08}
M.~de~Berg, O.~Cheong, M.~van Kreveld, and M.~Overmars.
\newblock Arrangements and duality.
\newblock In {\em Computational Geometry: Algorithms and Applications}, pages
  165--182. Springer, 2008.

\bibitem{efrat-05}
A.~Efrat.
\newblock The complexity of the union of $(\alpha, \beta)$-covered objects.
\newblock {\em SIAM J. Comput.}, 34:775--787, 2005.

\bibitem{gjs-ssnpcgp-76}
M.~Garey, D.~Johnson, and L.~Stockmeyer.
\newblock Some simplified {NP}-complete graph problems.
\newblock {\em Theoretical Computer Science}, 1(3):237 -- 267, 1976.

\bibitem{gkv-ipud-11}
M.~Gibson, G.~Kanade, and K.~Varadarajan.
\newblock On isolating points using disks.
\newblock In {\em Proc. ESA}, pages 61--69, 2011.

\bibitem{hclss-cebc-12}
S.~He, J.~Chen, X.~Li, X.~Shen, and Y.~Sun.
\newblock Cost-effective barrier coverage by mobile sensor networks.
\newblock In {\em Proc. INFOCOM}, pages 819--827, 2012.

\bibitem{h-mcnf-63}
T.~C. Hu.
\newblock Multi-commodity network flows.
\newblock {\em Operations Research}, 11(3):pp. 344--360, 1963.

\bibitem{klss-otcbrfr-13}
M.~Korman, M.~L\"offler, R.~I. Silveira, and D.~Strash.
\newblock On the complexity of barrier resilience for fat regions.
\newblock In {\em Proc. ALGOSENSORS}, pages 201--216, 2013.

\bibitem{kla-bcws-05}
S.~Kumar, T.-H. Lai, and A.~Arora.
\newblock Barrier coverage with wireless sensors.
\newblock In {\em Proc. MOBICOM}, pages 284--298, 2005.

\bibitem{kla-bcws-07}
S.~Kumar, T.-H. Lai, and A.~Arora.
\newblock Barrier coverage with wireless sensors.
\newblock {\em Wireless Networks}, 13(6):817--834, 2007.

\bibitem{m-pgsp-06}
D.~Marx.
\newblock Parameterized graph separation problems.
\newblock {\em Theoretical Computer Science}, 351(3):394--406, 2006.

\bibitem{mttv-sspnng-92}
G.~L. Miller, S.~Teng, W.~P. Thurston, and S.~A. Vavasis.
\newblock Separators for sphere-packings and nearest neighbor graphs.
\newblock {\em J. {ACM}}, 44(1):1--29, 1997.

\bibitem{PachS99}
J.~Pach and M.~Sharir.
\newblock On the boundary of the union of planar convex sets.
\newblock {\em Discrete {\&} Computational Geometry}, 21(3):321--328, 1999.

\bibitem{pv-psiuudnp-13}
R.~Penninger and I.~Vigan.
\newblock Point set isolation using unit disks is {NP}-complete.
\newblock {\em CoRR}, abs/1303.2779, 2013.

\bibitem{t-rwsn-11}
K.-C.~R. Tseng.
\newblock Resilience of wireless sensor networks.
\newblock Master's thesis, University of British Columbia, 2011.

\bibitem{tk-brsn-11}
K.-C.~R. Tseng and D.~Kirkpatrick.
\newblock On barrier resilience of sensor networks.
\newblock In {\em Proc. ALGOSENSORS}, pages 130--144, 2011.

\bibitem{whitesides-zhao-1990}
S.~Whitesides and R.~Zhao.
\newblock K-admissible collections of {Jordan} curves and offsets of circular
  arc figures.
\newblock Technical Report SOCS 90.08, School of Computer Science, McGill
  University, 1990.

\bibitem{x-sipamc-10}
M.~Xiao.
\newblock Simple and improved parameterized algorithms for multiterminal cuts.
\newblock {\em Theory of Computing Systems}, 46(4):723--736, 2010.

\bibitem{y-sppacgatm-12}
S.~Yang.
\newblock {\em Some Path Planning Algorithms in Computational Geometry and Air
  Traffic Management}.
\newblock PhD thesis, State University of New York, Stony Brook, 2012.

\end{thebibliography}

\appendix\clearpage

\end{document}